\documentclass[10pt,a4paper]{article}
\usepackage[utf8]{inputenc}
\usepackage[T1]{fontenc}
\usepackage[english]{babel}
\usepackage{palatino}
\usepackage{fullpage}
\usepackage{amsmath, amssymb, amsthm}
\usepackage{csquotes}
\usepackage{eurosym}
\usepackage{pifont}
\usepackage{verbatim}
\usepackage{float}
\usepackage{hyperref}
\usepackage{graphicx}
\usepackage{subfig}
\usepackage{tikz}
\usepackage{authblk}
\usetikzlibrary{arrows}
\usepackage[shortlabels]{enumitem}
\usepackage{underscore}
\usepackage{algorithmic}
\usepackage{algorithm}
\usepackage[most]{tcolorbox}

\newtheorem{theorem}{Theorem}
\newtheorem{question}{Question}
\newtheorem{lemma}[theorem]{Lemma}

\newtheorem{obs}[theorem]{Observation}

\newtheorem{case}{Case}

\numberwithin{subcase}{case}
\makeatletter
\@addtoreset{case}{theorem}
\makeatother

\newcommand{\source}{{\sf s}}
\newcommand{\target}{{\sf t}}

\renewcommand{\to}[1]{\overset{\text{#1}}{\rightsquigarrow}}
\newcommand{\reconf}[1]{\overset{\text{#1}}{\leftrightsquigarrow}}
\newcommand{\DSR}{DSR\textsubscript{TS}}
\newcommand{\DSRTJ}{DSR\textsubscript{TJ}}
\newcommand{\VCR}{VCR\textsubscript{TS}}
\newcommand{\ISR}{ISR\textsubscript{TS}}
\newcommand{\mno}{\emph{mno}}

\def\problem#1#2#3{%
\begin{tcolorbox}[enhanced,
  attach boxed title to top left={xshift=2mm,yshift=-3mm,yshifttext=-2mm},
  colback=gray!20!white, colframe=gray!70!black, colbacktitle=gray!70!black,
  boxrule=0.2mm, boxed title style={size=small}, title={#1}]
 \begin{description}[noitemsep,font=\textsf]
  \item[Instance:] {#2}
  \item[Question:] {#3}
 \end{description}
\end{tcolorbox}
}

\begin{document}

\title{Dominating sets reconfiguration under token sliding}

\author[1]{Marthe Bonamy}
\author[2]{Paul Dorbec}
\author[1]{Paul Ouvrard}
\affil[1]{Univ. Bordeaux, Bordeaux INP, CNRS, LaBRI, UMR5800, F-33400 Talence, France \thanks{\{marthe.bonamy, paul.ouvrard\}@u-bordeaux.fr}}
\affil[2]{Normandie Univ, UNICAEN, ENSICAEN, CNRS, GREYC, 14000 Caen, France \thanks{paul.dorbec@unicaen.fr}}

\date{}

\maketitle

\begin{abstract}
Let $G$ be a graph and $D_\source$ and $D_\target$ be two dominating sets of $G$ of size $k$. Does there exist a sequence  $\langle D_0 = D_\source, D_1, \ldots, D_{\ell-1}, D_\ell = D_\target \rangle$ of dominating sets of $G$ such that $D_{i+1}$ can be obtained from $D_i$ by replacing one vertex with one of its neighbors? In this paper, we investigate the complexity of this decision problem. We first prove that this problem is PSPACE-complete, even when restricted to split, bipartite or bounded treewidth graphs. On the other hand, we prove that it can be solved in polynomial time on dually chordal graphs (a superclass of both trees and interval graphs) or cographs.
\end{abstract}

\section{Introduction}

\paragraph{General introduction.}
Reconfiguration problems arise when, given an instance of a problem, we want to find a step-by-step transformation (called a \emph{reconfiguration sequence}) between two feasible solutions such that all intermediate solutions are also feasible. Unfortunately, such a transformation does not always exist and some solutions may even be \emph{frozen}, i.e., they can not be modified at all. In this context, two natural questions arise: (i) When can we ensure that there exists such a transformation? (ii) What is the complexity of deciding whether a reconfiguration sequence exists?

Interest in combinatorial reconfiguration steadily increased during the last decade. Reconfiguration of several problems, including \textsc{Coloring} \cite{BONAMY2019179,CERECEDA2008913,FEGHALI2019169}, \textsc{Independent Set} \cite{Bonsma16,BMP17,LokshtanovM18}, \textsc{Dominating Set} \cite{Haddadan:2016:CDS:3010167.3010682,LOKSHTANOV2018122,Mouawad2017,SMN14} and \textsc{Satisfiability} \cite{Gopalan:2009:CBS:1654348.1654357, MNR17} have been studied. For an overview of recent results on reconfiguration problems, the reader is referred to the surveys of van den Heuvel \cite{Heuvel13} and Nishimura \cite{Nishimura18}. In this article, we focus on the reconfiguration of dominating sets.

A dominating set is a set of vertices such that every vertex not in the set has a neighbor in it.
One can represent a dominating set as a set of tokens, where exactly one token is placed on each vertex that is part of the dominating set.
Then, modifying a dominating set corresponds to shifting the tokens according to some rule, called a {\em reconfiguration rule}. In the literature, three kinds of operations have been mainly studied:

\begin{enumerate}
    \item Token Addition and Removal ({\sf TAR}): one can add or remove a token as long as the total number of tokens does not go beyond a given threshold;
    \item Token Jumping ({\sf TJ}): one can move a token to any vertex of the graph;
    \item Token Sliding ({\sf TS}): one can slide a token along an edge, i.e., one moves a token to a neighbor of its current location.
\end{enumerate}

One can observe that in the last two models, the size of each solution remains constant at any time, as opposed to what happens in the {\sf TAR} model. In this article, we are mostly interested in the Token Sliding model.

We define the reconfiguration graph for domination, denoted $\mathcal{R}_k(G)$ as follows: the vertices of $\mathcal{R}_k(G)$ are the dominating sets of size $k$ and there is an edge between two vertices if and only if one can go from the first to the second thanks to the reconfiguration rule that we consider (token sliding in our case). Three natural problems can be identified:

\begin{enumerate}
    \item  The \emph{reachability} problem: given a graph $G$ and two dominating sets $D_\source$ and $D_\target$, is there a path between $D_\source$ and $D_\target$ in $\mathcal{R}_k(G)$? In other words, does there exist a reconfiguration sequence between $D_\source$ and $D_\target$?

    \item The \emph{connectivity} problem: given a graph $G$, is the reconfiguration graph $\mathcal{R}_k(G)$ connected?

    \item The \emph{shortest path} problem: given a graph $G$, two dominating sets $D_\source$ and $D_\target$ and an integer $\ell$, is the distance in $\mathcal{R}_k(G)$ between $D_\source$ and $D_\target$ at most $\ell$? In other words, does there exist a reconfiguration sequence between $D_\source$ and $D_\target$ of length at most $\ell$?
\end{enumerate}

In this article, we focus on the reachability version and we will denote this problem by \DSR\ for short.
We adopt the same notation for \textsc{Vertex Cover Reconfiguration} and \textsc{Independent Set Reconfiguration} (the reachability question under the token sliding rule) and denote these two problems by \VCR\ and \ISR, respectively.

\paragraph{Related results.}
The reconfiguration of dominating sets has been mainly studied under the \emph{Token Addition and Removal} model. Haas and Seyffarth gave sufficient conditions to guarantee the connectivity of the reconfiguration graph according to $k$, the cardinality threshold of dominating sets \cite{Haas2014}. More precisely, they proved that $\mathcal{R}_{n-1}(G)$ (where $n$ is the number of vertices of $G$) is connected if $G$ has at least two independent edges. This value can be lowered to $\Gamma$+1 (where $\Gamma$ is the maximum size of a minimal dominating set) if the input graph $G$ is chordal or bipartite. Suzuki et al.~\cite{SMN14} showed that this result cannot be generalized to any graph since they constructed an infinite family of graphs for which $\mathcal{R}_{\Gamma+1}(G)$ is not connected. On the positive side, they proved that $\mathcal{R}_{n - \mu}(G)$ is connected if $G$ has a matching of size $\mu +1$.

Haddadan et al.~\cite{Haddadan:2016:CDS:3010167.3010682} studied the complexity of the reconfiguration of dominating sets under the token addition and removal rule from a graph classes perspective. They proved that the reachability problem is PSPACE-complete, even if the input graph is a split graph, a bipartite graph, has bounded bandwidth or is planar with maximum degree six. On the other hand, they gave linear-time algorithms for trees, interval graphs or cographs.

Mouawad et al.~\cite{Mouawad2017} studied the parameterized complexity of \textsc{Dominating Set Reconfiguration} under token addition and removal. They proved that this problem is W[2]-hard when parameterized by $k+\ell$, where $k$ is the threshold and $\ell$ the length of the reconfiguration sequence. As a positive result, Lokshtanov et al.~\cite{LOKSHTANOV2018122} gave a fixed-parameter algorithm with respect to $k$ for graphs excluding $K_{d,d}$ as a subgraph, for any constant $d$.

The third author also considered this problem (still in the {\sf TAR} model) through the lens of an \emph{optimization variant} (see Blanch\'e et al. \cite{Blanche20}) as recently introduced by Ito et al. for independent sets~\cite{ItoMNS19}.

To the best of our knowledge, the reconfiguration of dominating sets under {\sf TS} has only been studied from a structural perspective. Fricke et al.~\cite{FrickeHHH11} introduced the concept of $\gamma$-graph which corresponds to the reconfiguration graph $\mathcal{R}_{\gamma}(G)$ under the {\em token sliding} rule. In particular, they proved that $\mathcal{R}_{\gamma}(G)$ is connected and bipartite if $G$ is a tree. For a more complete overview, the reader is referred to~\cite{mynhardt2020reconfiguration}. 

\paragraph{Our contribution.}

\begin{figure}[bt]
    \centering
    \includegraphics[width=0.8\textwidth]{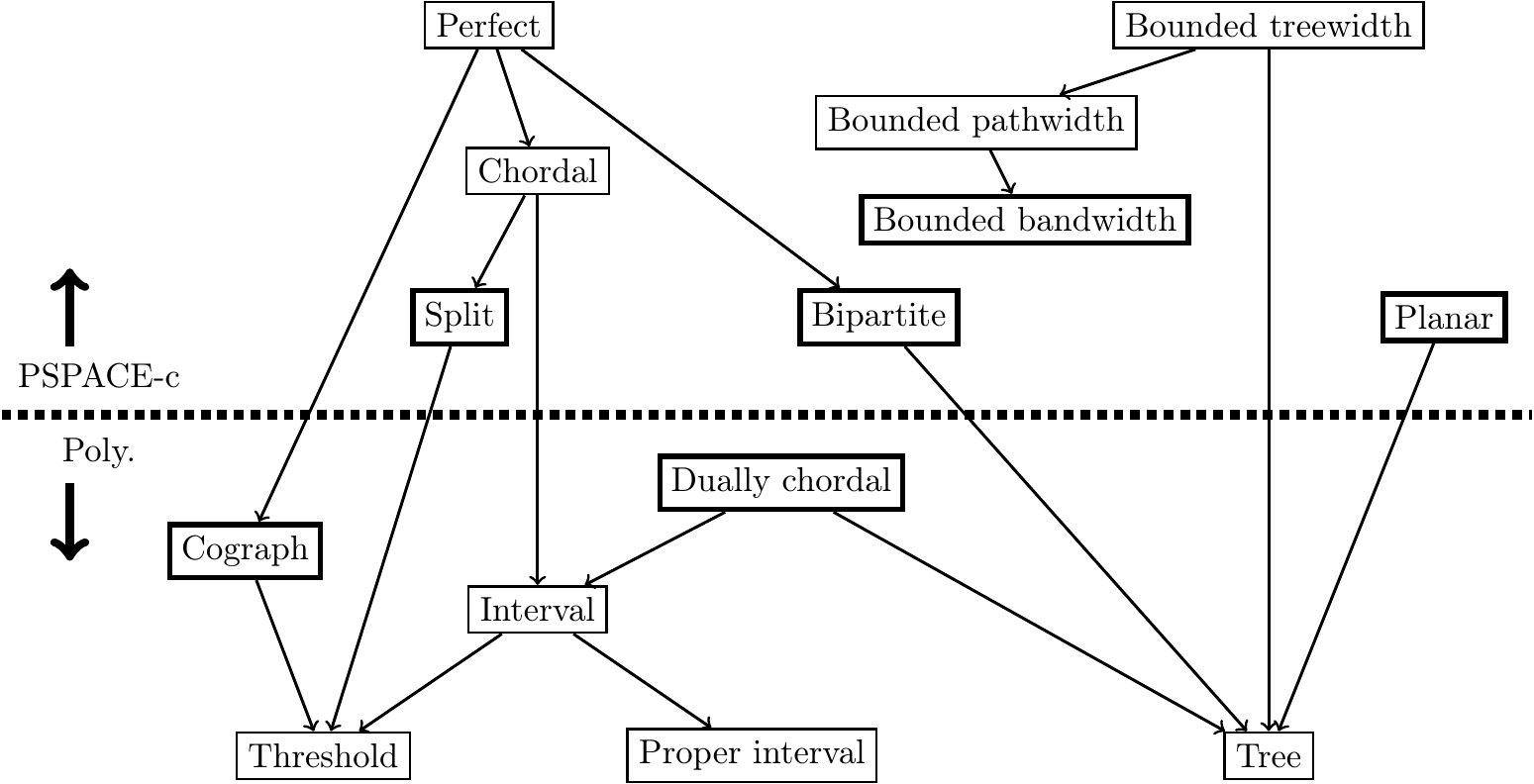}
    \caption{Our results: the frontier between PSPACE-completeness and tractability.}
    \label{fig:summary}
\end{figure}

In this article, we are interested in the reachability question of dominating sets reconfiguration under token sliding. This reconfiguration rule has already been studied for various reconfiguration problems but not for dominating sets, to the best of the authors' knowledge. 

We tackle this problem with a complexity perspective according to several graph classes: in Section \ref{sec:pspace}, we prove for instance that \DSR\ is PSPACE-complete for split graphs or bipartite graphs.
Note that the reductions used in the proofs of Theorems~\ref{thm:bipartite} and \ref{thm:planar-bandwidth} are identical to the ones of \cite{Haddadan:2016:CDS:3010167.3010682}, which are quite standard (see, e.g.,~\cite{BERTOSSI198437}). Our reduction to prove the PSPACE-completeness of \DSR\ for split graphs is similar but not identical to the one of~\cite{Haddadan:2016:CDS:3010167.3010682}, as we reduced from \DSRTJ\ and not from {\sc Vertex Cover Configuration}.

In Section \ref{sec:poly}, we show that this problem can be solved in polynomial time on other graph classes such as cographs or dually chordal graphs (the formal definitions of these two graph classes is given in Subsections~\ref{join-cographs} and \ref{dually}, respectively). Note that our result for cographs is a consequence of Theorem~\ref{thm:join} which is more general (see the discussion in Subsection~\ref{join-cographs}). Note also that our result on dually chordal graphs generalizes the ones of \cite{Haddadan:2016:CDS:3010167.3010682} for trees and interval graphs since the class of dually chordal graphs is a superclass of both interval graphs and trees as discussed in Subsection~\ref{dually}.

Figure \ref{fig:summary} gives an overview of our results where $A \rightarrow B$ means that the class $B$ is properly included in the class $A$.

\section{Preliminaries}

This section is devoted to some basic definitions of graph theory used in this article, followed by a more formal introduction of the problem we are interested in.

Each graph $G=(V,E)$ considered is simple (i.e., $G$ is undirected and has no multiple edges or loops) where $V$ represents the vertex set of $G$ and $E$ its edge set. We denote by $n = |V|$ and $m = |E|$ the number of vertices and edges of $G$. The \emph{eccentricity} of a vertex $u$ denoted by $\epsilon(u)$ is the maximum distance between $u$ and any other vertex. For a subset of vertices $S \subseteq V$, we denote by $G[S]$ the subgraph induced by $S$. For a vertex $u \in V$, we denote by $N_G(u)$ its open neighborhood, i.e., the set $\{v\ |\ uv \in E\}$ and by $N_G[u]$ its \emph{closed neighborhood}, i.e.\ the set $N_G(u) \cup \{u\}$. For a subset of vertices $S \subseteq V$, we define the closed neighborhood of $S$ as the union of the closed neighborhood of the vertices in $S$, i.e., $N_G[S] = \bigcup_{u \in S} N_G[u]$.

Let $G_1$ and $G_2$ be two graphs. We recall two basic binary operations on graphs: the disjoint union and the join operations. The disjoint union $G_1 \cup G_2$ of two graphs on disjoint vertex sets is the graph with vertex set $V(G_1 \cup G_2) = V(G_1) \cup V(G_2)$ and edge set $E(G_1 \cup G_2) = E(G_1) \cup E(G_2)$. The join operation can be obtained from the disjoint union by adding all possible edges between $G_1$ and $G_2$. More formally, the join of $G_1$ and $G_2$ denoted by $G_1 + G_2$ is the following graph:

\begin{itemize}
    \item $V(G_1 + G_2) = V(G_1) \cup V(G_2)$;
    \item $E(G_1 + G_2) =  E(G_1) \cup E(G_2) \cup \{uv\ |\ u \in V(G_1), v \in V(G_2)\}$.
\end{itemize}

A dominating set for a graph $G=(V,E)$ is a subset of vertices $D \subseteq V$ such that $N[D] = V$, i.e., each vertex either belongs to $D$ or has a neighbor in $D$. For a graph $G$, we denote by $\gamma(G)$ the domination number of $G$ defined as the minimum size of a dominating set. Let $G$ be a graph and $D$ a dominating set of $G$. We say that $u$ is a private neighbor of $v$ (with respect to $D$) if $u \not\in D$ and $v$ is the only neighbor of $u$ in $D$. Therefore, a dominating set is inclusion-wise minimal if and only if each of its vertices has a private neighbor.

\paragraph{Our problem.}
In the token sliding model, a natural question is whether we should authorize more than one token to be placed on a vertex during the reconfiguration sequence. Here is an example where it makes a difference: consider the star graph $S_n$ on $n+1$ vertices and two dominating sets $D_1$ and $D_2$ of $S_n$ of size $k$, with $k \in [2,n-1]$. Any dominating set of that size necessarily contains the central vertex. To reconfigure $D_1$ into $D_2$, we are forced to move a token from one leaf to another, which can only be done by going through the central vertex which already contains a token. Given such artificially negative examples, we choose to allow the superposition of tokens on a vertex. Note that this question did not arise in previous papers considering the token sliding model, to the best of our knowledge. Indeed, for problems like independent set, there can be no question of superposing two tokens, as two tokens cannot be adjacent in the first place. In the aforementioned paper considering token sliding for dominating sets, they exclusively consider that model in the case of minimum dominating sets: if superposition was an option, there would be a smaller dominating set, which is impossible.

Let $G$ be a graph, $D_\source$ and $D_\target$ be two dominating sets of $G$ of same size $k$. We say that $D_\source$ is reconfigurable into $D_\target$ by token sliding if there exists a sequence $S = \langle D_0 = D_\source, D_1, \ldots, D_{\ell-1}, D_\ell = D_\target \rangle$ that satisfies the two following properties:

\begin{itemize}
    \item each $D_i$ is a multiset of size $k$ that is a dominating set of $G$;
    \item there exists an edge $uv$ such that $D_{i+1} = D_i \setminus \{u\} \cup \{v\}$, i.e., we slide the token placed on the vertex $u$ along the edge $uv$. We denote this move by $u \to{{\sf TS}} v$.
\end{itemize}

We call such a sequence a {\sf TS}-sequence and we denote this property by $D_\source \to{{\sf TS}} D_\target$. We also say that $(G, D_\source, D_\target)$ is a \textsf{yes}-instance for the \DSR\ problem. 

We also introduce the two following notation, where $D_\source$ and $D_\target$ are two dominating sets of $G$ of size $k$.

\begin{itemize}
    \item $D_\source \reconf{{\sf TAR}} D_\target$: one can reconfigure $D_\source$ into $D_\target$ under the {\sf TAR} model; each intermediate solution is of size at most $k+1$;
    \item $D_\source \reconf{{\sf TJ}} D_\target$: one can reconfigure $D_\source$ into $D_\target$ under the {\sf TJ} model; each intermediate solution is of size exactly $k$.
\end{itemize}

A useful observation is that each reconfiguration sequence (and thus in particular a {\sf TS}-sequence) is reversible: if $D_\source \to{{\sf TS}} D_\target$ holds, then $D_\target \to{{\sf TS}} D_\source$ holds too. We thus denote this relation by $D_\source \reconf{{\sf TS}} D_\target$. Figure \ref{fig:ts-sequence} gives an example of a {\sf TS}-sequence.

\begin{figure}[bt]
    \centering
    \includegraphics[width=0.9\textwidth]{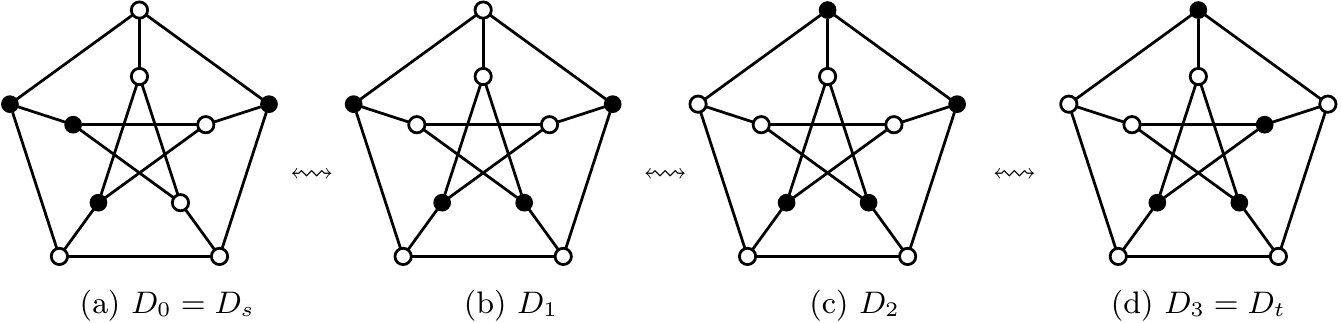}
    \caption{Example of {\sf TS}-sequence from $D_\source$ to $D_\target$.}
    \label{fig:ts-sequence}
\end{figure}

We are now ready to define properly the \textsc{Dominating Set Reconfiguration} problem under token sliding.

\medskip

\problem{\DSR }%
    {A graph $G=(V,E)$ and two dominating sets $D_\source$ and $D_\target$ of cardinality $k$ of $G$.}%
    {Is there a {\sf TS}-sequence between $D_\source$ and $D_\target$, i.e., does $D_\source \reconf{{\sf TS}} D_\target$ hold? }

\medskip    
    
We end this section by the following observation, showing that being reconfigurable is not a monotone property.

\begin{theorem}
For every $\ell \geq 3$, there exists a graph $G_\ell$ where, for every $k < \ell$, every dominating set of size $k$ can be reconfigured into any other, while there are two dominating sets of size $\ell$ that cannot be reconfigured one into the other.
\end{theorem}

\begin{proof}
We first prove the statement for $k=2$.
For every integer $\ell > 2$, we define the graph $G_\ell$ such that $G_\ell$ contains exactly one dominating set of size $\gamma(G) = 2$ but for which the dominating sets of size $\ell$ are not reconfigurable. To construct $G_\ell$, we first create $\ell$ pairs of triangles $\{(G^i_1, G^i_2), \ldots, (G^\ell_1, G^\ell_2)\}$ such that $G^i_1$ and $G^i_2$ share exactly one vertex $w_i$. Moreover, let all the $G^i_1$'s share a vertex $u$ and all the $G^i_2$'s share a vertex $v$ (see Figure \ref{fig:G3} for $G_3$ as an example). Note that we have $\gamma(G_\ell) = 2$ since $\{u,v\}$ is a dominating set and $G_\ell$ does not contain a universal vertex (i.e., a vertex adjacent to all the other vertices). 

Consider the dominating set $D_\source = \{w_1, \ldots, w_\ell\}$. It is a dominating set of $G_\ell$ of size $\ell$. By token sliding, $D_\source$ cannot be reconfigured into any other dominating set of size $\ell$. Indeed, in $D_\source$ we cannot move any $w_i$ in a triangle because it would leave the other triangle of the pair $(G^i_1, G^i_2)$ not dominated. Note that any set of $\ell$ vertices containing $u$ and $v$ is a dominating set of $G_\ell$, hence the existence of two dominating sets of size $\ell$ as desired.

Consider now $k < \ell$. Any dominating set of $G_\ell$ on fewer than $\ell$ vertices contains both $u$ and $v$. Indeed, if for instance $u$ is not in the dominating set, then $\ell$ extra vertices are necessary to dominate the triangles $G^i_1$. Therefore, any dominating set $D$ of $G_\ell$ on $k$ vertices contains both $u$ and $v$. The other vertices are therefore not necessary for domination purposes, and we can slide them around as desired, superposing them with $u$ and $v$ arbitrarily. There are many dominating sets on $k$ vertices, but they all contain $u$ and $v$ and can be trivially reconfigured one into another. \qedhere
\end{proof}

\begin{figure}[bt]
    \centering
    \includegraphics[width=0.35\textwidth]{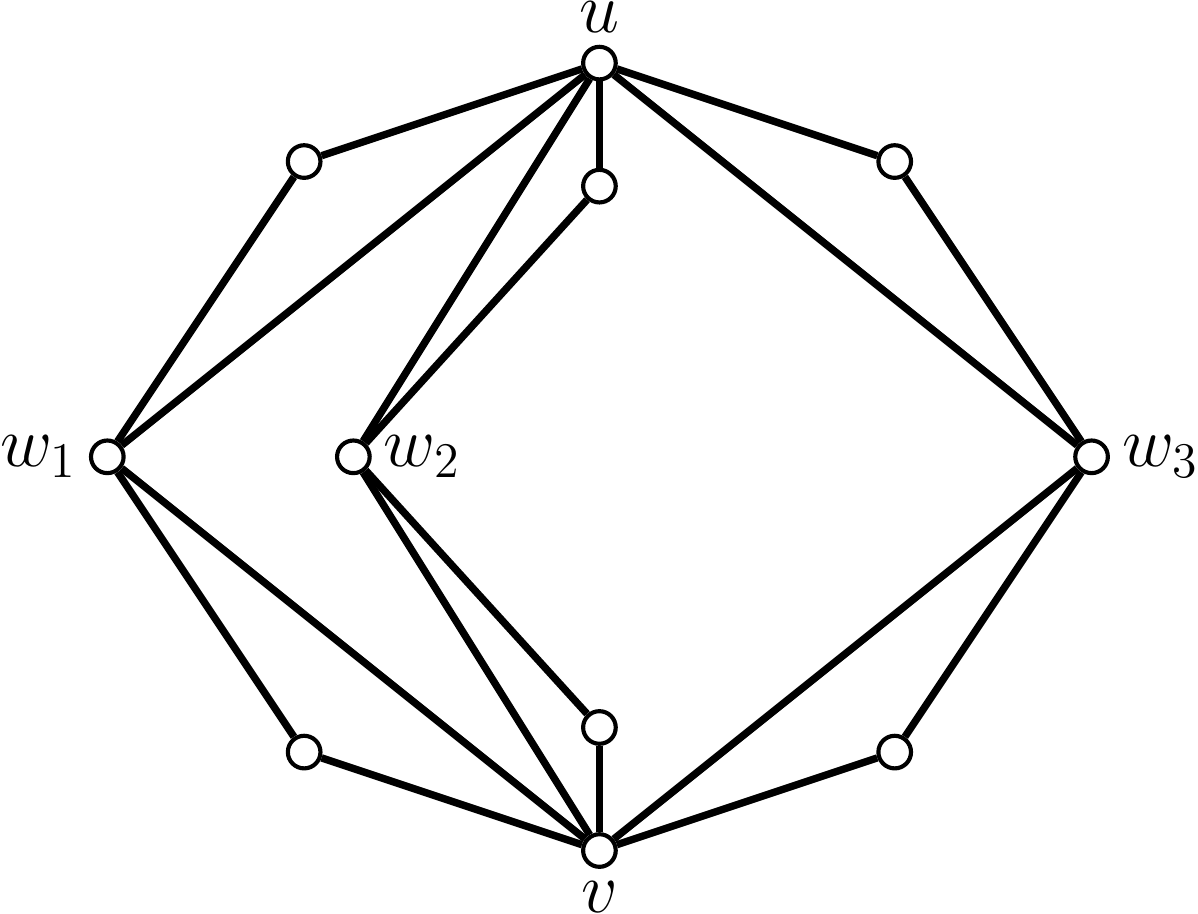}
    \caption{Graph $G_3$.}
    \label{fig:G3}
\end{figure}

\section{PSPACE-completeness} \label{sec:pspace}

In this section, we study the complexity of \DSR\ in the general case. We show that this problem is PSPACE-complete, even when restricted to split graphs, bipartite graphs or bounded treewidth graphs. Let us first recall the following result from Haddadan et al.~\cite{Haddadan:2016:CDS:3010167.3010682}, stating the complexity of the reconfiguration problem for the {\sf TAR} model.

\begin{theorem}[\cite{Haddadan:2016:CDS:3010167.3010682}] \label{thm:TAR-PSPACE}
Let $G$ be a graph and $D_\source, D_\target$ be two dominating sets of $G$ of size $k$. Deciding whether $D_\source \reconf{{\sf TAR}} D_\target$ is PSPACE-complete.
\end{theorem}

Note that the problem remains PSPACE-complete, even if the input graph is a planar graph with maximum degree 6, has bounded bandwidth, is bipartite or is a  split graph as pointed out previously.
Let us now show that the {\sf TJ} and {\sf TAR} rules are equivalent under some constraints. Note that a similar proof can be found in~\cite{Haas2014}.

\begin{lemma} \label{lemma:equivalenceTAR-TJ}
Let $G$ be a graph and $D_\source$ and $D_\target$ be two dominating sets of $G$ of size $k$. We have $D_\source \reconf{{\sf TAR}} D_\target$ if and only if $D_\source \reconf{{\sf TJ}} D_\target$.
\end{lemma}

\begin{proof}
The proof is an adaptation of the Theorem 1 of Kami\'nski et al. \cite{KAMINSKI20129}. 

($\Leftarrow$)
Suppose first $D_\source \reconf{{\sf TJ}} D_\target$, and let $S$ be a {\sf TJ}-sequence that reconfigures $D_\source$ into $D_\target$. This sequence corresponds to a sequence of moves $u \to{{\sf TJ}} v$. We construct a {\sf TAR}-sequence by replacing each atomic move $u \to{{\sf TJ}} v$ with two moves of the {\sf TAR} model: we first add $v$ and then delete $u$. By first adding $v$, we preserve the domination property. Besides, since we immediately delete $u$ after the addition of $v$, each intermediate solution is of size at most $k+1$, as desired.

($\Rightarrow$)
For the other direction, let $S'$ be a {\sf TAR}-sequence that reconfigures $D_\source$ into $D_\target$. Note that since $|D_\source| = |D_\target| = k$, $S'$ is of even length. Moreover, by hypothesis, $S'$ does not contain a configuration of size more than $k+1$. If all the configurations of $S'$ are of size $k$ or $k+1$, this means that $S'$ corresponds to an alternation of an addition of a token on some vertex $v$ immediately followed by the deletion of a token on a vertex $u$. Therefore, to get a {\sf TJ}-sequence, we simply replace each of these subsequences by a move $u \to{{\sf TJ}} v$. Suppose now that $S'$ contains some configuration of size less than $k$ and consider a configuration, let us say $D_i$, of smallest size. Since $D_i$ is a configuration of smallest size, this means that it has been obtained from $D_{i-1}$ by the deletion of some vertex $x$. We also know that the configuration $D_{i+1}$ is obtained from $D_i$ by the addition of some vertex $y$. If $x = y$, then these two steps are redundant and can simply be ignored. Otherwise, observe that if we first add $y$ and then delete $x$, the new sequence is still valid. If all the configurations are of size $k$ or $k+1$, we immediately obtain a {\sf TJ}-sequence. Otherwise, we can repeat this process until this is the case.
\end{proof}

As a corollary of Theorem \ref{thm:TAR-PSPACE} and Lemma \ref{lemma:equivalenceTAR-TJ}, we obtain that deciding whether two dominating sets of size $k$ of a graph $G$ can be reconfigured under the token jumping model is a PSPACE-complete problem. We are now ready to prove Theorem \ref{thm:split}.

\begin{theorem} \label{thm:split}
\DSR\ is PSPACE-complete on split graphs.
\end{theorem}

\begin{proof}
First, note that the problem is in PSPACE \cite{ITO20111054}. We give a polynomial-time reduction from \DSRTJ, which is PSPACE-complete as discussed above. Let $G=(V,E)$ be a graph with $V(G) = \{v_1, \ldots, v_n\}$. We construct the corresponding split graph $G'$ as follows:

\begin{itemize}
    \item $V(G') = V_1 \cup V_2$ where $V_1 = \{v_1, \ldots, v_n\}$ and $V_2 = \{w_1, \ldots, w_n\}$;
    \item $E(G') = \{uv\ |\ u,v \in V_1\} ~ \cup ~ \{v_iw_j\ |\  v_j \in N_G[v_i]\}$, i.e., we add all possible edges in $V_1$ so that $V_1$ forms a clique. We also add an edge between a vertex $v_i \in V_1$ and a vertex $w_j \in V_2$ if and only if the corresponding vertex $v_j$ in the original graph $G$ belongs to the closed neighborhood of $v_i$ in $G$.
\end{itemize}

\begin{figure}[bt]
    \centering
    \includegraphics[width=0.4\textwidth]{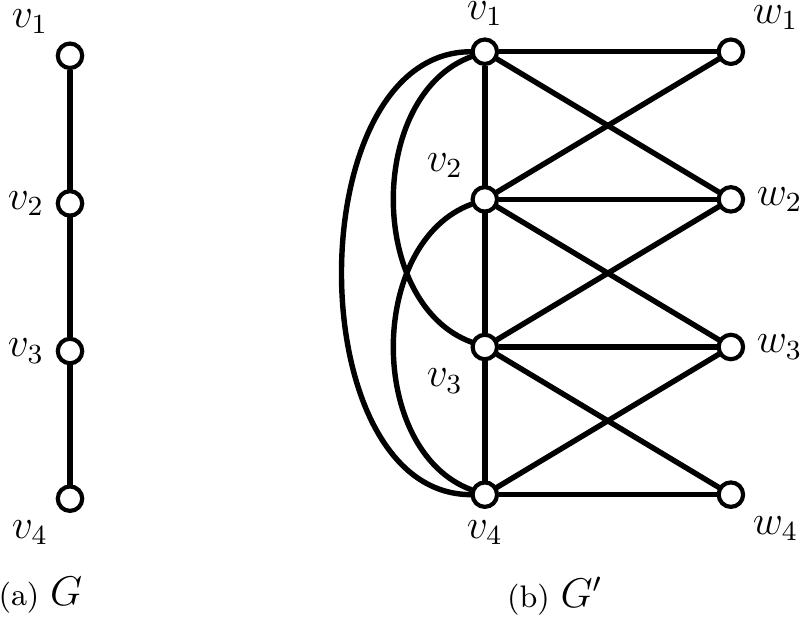}
    \caption{Example for the reduction of Theorem~\ref{thm:split}.}
    \label{fig:reduc-split}
\end{figure}

Observe that $G'$ is a split graph since $V_1$ forms a clique and $V_2$ an independent set (see Figure \ref{fig:reduc-split} for an example).
To a set of vertices of $G$, we associate the corresponding vertices of $V_1$ in $G'$. By definition of $G'$, any dominating set $D$ of $G$ is also a dominating set for $G'$: indeed, a vertex $v_i \in V_1$ dominates all the vertices in $V_1$ (since it is a clique) and all the vertices in $V_2$ that correspond to vertices in its closed neighborhood in $G$. That $D$ dominates $G$ allows us to conclude that the corresponding set also dominates $V_2$. Hence, $D$ is also a dominating set of $G'$.

Let ($G$, $D_\source, D_\target$) be an instance of \DSRTJ; we reduce this instance to the instance of \DSR\ $(G',D_\source, D_\target)$. This reduction can be done in quadratic time.
Now, we need to prove that $D_\source \reconf{{\sf TS}} D_\target$ if and only if there is a reconfiguration sequence between $D_\source$ and $D_\target$ in $G'$ using the token jumping model. 

($\Leftarrow$)
Consider a {\sf TJ}-sequence in $G$, and translate it to $G'$. All intermediate sets still are dominating sets, and since all pairs of vertices are joined by an edge in $V_1$, this sequence is a valid {\sf TS}-sequence in $G'$.

($\Rightarrow$)
We now prove the other direction.
Let $\langle D_0 = D_\source, \ldots, D_p = D_\target \rangle$ be a {\sf TS}-sequence in $G'$. First, observe that any dominating $D'$ of $G'$ such that $D' \subseteq V_1$ corresponds to a dominating set of $G$. Indeed, any vertex $w_j \in V_2$ is dominated by a vertex $v_i \in V_1$ and by construction of $G'$, $v_iv_j \in E(G)$. Hence, $v_j$ is dominated  by $v_i$ and thus $D'$ is also a dominating set of $G$. Hence, if the sequence does not use vertices in $V_2$, we immediately obtain a {\sf TJ}-sequence in $G$ from $D_\source$ to $D_\target$, as the token jumping model does not require adjacency.
Suppose on the other hand that the sequence goes through some vertices in $V_2$. Since all vertices are initially in $V_1$, there is a subsequence that contains a move $v_i \to{{\sf TS}} w_j$.
Since $w_j \notin V_1$, there exists a later step where the token on $w_j$ is moved to an adjacent vertex $v_k$ in $V_1$ (since $V_2$ is independent). However, $w_j$ does not dominate any vertex in $V_2$ (since $V_2$ is a stable set) and thus $N[w_j] \subseteq N[v_k]$. Therefore, we simply replace these two moves by a single move $v_i \to{{\sf TS}} v_k$. We can thus assume that the reconfiguration sequence only uses vertices in $V_1$. The conclusion follows. \qedhere
\end{proof}

Next, we prove that \DSR\ is PSPACE-complete on bipartite graphs.
We use a reduction from  the \textsc{Vertex Cover Reconfiguration} problem under token sliding (or \VCR\ for short). Recall that a vertex cover is a set of vertices such that every edge has an endpoint in the set. The complement of a vertex cover is an independent set whose reconfiguration is known to be PSPACE-complete on planar graphs of maximum degree 3 \cite{Hearn:2005:PSP:1140710.1140715, BONSMA20095215} or on bounded bandwidth graphs \cite{WROCHNA20181}. Hence, \VCR\ is PSPACE-complete, even when restricted to these two classes.

\begin{theorem} \label{thm:bipartite}
\DSR\ is PSPACE-complete on bipartite graphs.
\end{theorem}

\begin{figure}[b]
    \centering
    \includegraphics[width=0.65\textwidth]{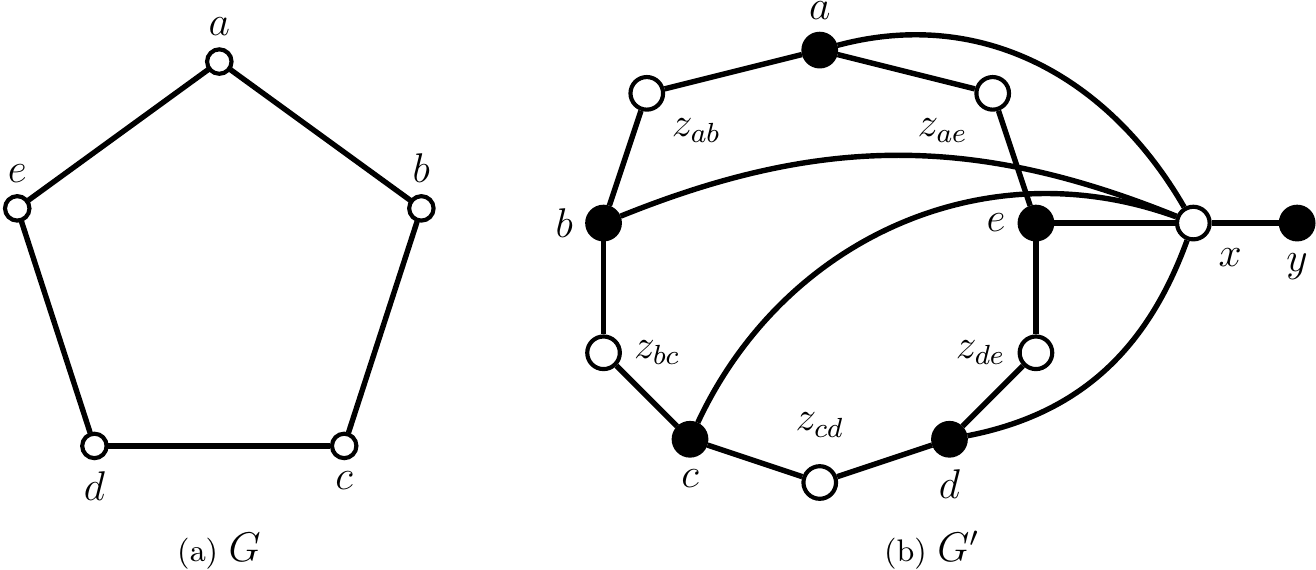}
    \caption{Example for the reduction of Theorem \ref{thm:bipartite}.}
    \label{fig:reduc-bipartite}
\end{figure}

\begin{proof}
We give a polynomial-time reduction from \VCR. This is an adaptation of the well-known reduction from \textsc{Vertex Cover} to \textsc{Dominating Set} \cite{Garey:1990:CIG:574848}. Let $G=(V,E)$ be a graph. We construct the corresponding bipartite graph $G'=(V_1 \uplus V_2, E')$ as follows: for each edge $uv \in E$, add $u$ and $v$ to $V_1$ and a new vertex $z_{uv}$ of degree two to $V_2$ that is adjacent to exactly $u$ and $v$. Note that $E'$ does not contain the edge $uv$ so that $V_1$ induces an independent set. Finally, add to $V_2$ a vertex $x$ adjacent to all the vertices in $V_1$ and attach to $x$ a degree-one vertex $y$ which is added to $V_1$ (see Figure \ref{fig:reduc-bipartite} for an example). Formally, the graph $G'$ is the following:

\begin{itemize}
    \item $V(G') = V_1 \cup V_2$ where $V_1 = V(G) \cup \{y\}$ and $V_2 = \{z_{uv}\ |\ uv \in E\} \cup \{x\}$;
    \item $E' = \{uz_{uv} \textrm{ and } z_{uv}v\ |\ u,v \in V_1 \textrm{ and } z_{uv} \in V_2\} \cup \{ xv\ |\ v \in V_1\} \cup \{xy\}$.
\end{itemize}

Observe that $G'$ is bipartite and the reduction can be done in polynomial time. We now prove that the vertex covers of $G$ of size $k$ are reconfigurable if and only if the dominating sets of $G'$ of size $k+1$ are. Let $(G,C_s,C_\target)$ be an instance for the \VCR\ problem.
We define the corresponding instance for the \DSR\ problem as $(G',C_s \cup \{x\},C_\target \cup \{x\})$. Since $C_s$ is a vertex cover of $G$, for every edge $uv \in E(G)$ we have $\{u,v\} \cap C_s \neq \emptyset$ and thus the vertices $u,v,z_{uv}$ are dominated by $C_s$ in $G'$. Now $x$ dominates both $x$ and $y$, so  $D_\source = C_s \cup \{x\}$ is a dominating set of $G'$, and by the same argument, so is $D_\target = C_\target \cup \{x\}$. Since \VCR\ and \DSR\ both employ the same reconfiguration rule, we simply denote by $u \to{} v$ (instead of $u \to{\sf TS} v$) a move of a reconfiguration sequence between $C_\source$ and $C_\target$ (respectively $D_\source$ and $D_\target$).

($\Rightarrow$)
We start with the only if direction. First, it immediately follows from the definition of $D_\source$ and $D_\target$ that $D_\source \setminus \{x\} = C_s$ and $D_\target \setminus \{x\} = C_\target$. Let us assume that $(G,C_s,C_\target)$ is a \textsf{yes}-instance for the \VCR\ problem. Then, there exists a reconfiguration sequence $S$ using the token sliding model between $C_s$ and $C_\target$. One can construct a sequence $S'$ for $G'$ by replacing a move $u \to{} v$ (where $uv \in E(G))$ of $S$ into two moves: $u \to{} z_{uv}$ followed by $z_{uv} \to{} v$. We need to prove that the domination property is preserved at every step. First, observe that each intermediate solution contains $x$, so each move of the form $z_{uv} \to{} v$ is safe because $u$ is still dominated by $x$ and $z_{uv}$ by $v$. Therefore, the only risk is to leave some vertex $z_{wu}$ non dominated after a move $u \to{} z_{uv}$. In that case, this implies that $w$ does not belong to the solution, which in turn means that the edges $wu$ and $uv$ are covered only by $u$. Therefore, the move $u \to{} v$ of the sequence $S$ is not valid (because the edge $wu$ is no longer covered), a contradiction. Therefore, $(G', D_\source,D_\target)$ is a \textsf{yes}-instance for the \DSR\ problem.

($\Leftarrow$)
It remains to prove the if direction. Suppose that $(G',D_\source,D_\target)$ is a \textsf{yes}-instance for the \DSR\ problem. Then, there exists a reconfiguration sequence $S' = \langle D_\source, \ldots, D_\target \rangle$ in $G'$. First, observe that at each step, $y$ needs to be dominated and thus either $x$ or $y$ belongs to each solution. Moreover, initially, $D_\source$ does not contain $y$. We can ignore all moves of the form $x \to{} y$ (each such move will be eventually followed by a move $y \to{} x$), and assume that $x$ contains at least one token in each solution. Therefore, the only vertices whose domination is not immediate by the existence of a token on $x$ are the vertices of the form $z_{uv}$, i.e., the vertices that correspond to the edges of $G$.
Recall that $D_\source = C_\source \cup \{x\}$, so every vertex in $D_\source \setminus \{x\}$ belongs to $V(G') \cap V(G)$. We consider in turn the two other possible moves $u \to{} v$, where $u \in V(G) \cap V(G')$ (i.e., $u$ corresponds to a vertex of the original graph $G$), and $v$ either belongs to $V(G') \setminus V(G)$ or $v=x$. We focus on the next operation (which may not be consecutive) that touches the vertex $v$.
Suppose first that $v \in V(G') \setminus V(G)$, i.e., $v$ corresponds to a vertex $z_{uu'}$ for some vertex $u' \in N_G[u]$. If the next move that touches $z_{uu'}$ is $z_{uu'} \to{} u$, these two operations can be ignored. Otherwise, since $z_{uu'}$ has degree two, the next operation that touches $z_{uu'}$ is $z_{uu'} \to{} u'$. Moreover, we claim that we can assume that $z_{uu'} \to{} v$ is the operation that immediately follows the move $u \to{} z_{uu'}$. Indeed, $N_{G'}[z_{uu'}] \subseteq N_{G'}[u]$ so if a dominating set $D$ contains $z_{uu'}$, $D' = (D \setminus \{z_{uu'}\}) \cup \{u'\}$ is also a dominating set of $G'$. So one can assume that in $S'$, if we have a move $u \to{} z_{uu'}$, it will be immediately followed by a move $z_{uu'} \to{} u'$. In that case, one can replace these two moves by $u \to{} u'$ in a reconfiguration sequence from $C_\source$ to $C_\target$. Let us now consider the other possible move: $u \to{} x$. If the next move that touches $x$ is $x \to{} u$, we again simply ignore these two steps. Let $D_i$ be the dominating set of $S'$ to which the move $u \to{} x$ is applied. Recall that when a token is moved from a vertex $a$ to a vertex $z_{ab}$ (for some neighbor $b$ of $a$), it is followed by $z_{ab} \to{} b$. Therefore, we know that $D_i$ does not contain any vertex of the form $z_{uv}$. So for every edge $uu'$ incident to $u$, $D_{i+1}$ must contain $u$ (this is possible if $u$ has at least two tokens in $D_i$) or $u'$ in order to dominate $z_{uu'}$. Hence, $C_{i+1} = D_{i+1} \setminus \{x\}$ is a vertex cover of $G$. If the next move that touches $x$ is $x \to{} u'$, one can safely replace these two moves $u \to{} x$ and $x \to{} u'$ by $d$ moves where $d$ is the distance between $u$ and $u'$ in $G$.

Therefore, one can obtain from $S'$ a {\sf TS}-sequence that reconfigures $C_s$ into $C_\target$ and thus $(G,C_s,C_\target)$ is a \textsf{yes}-instance for \VCR, as desired. This concludes the proof of Theorem~\ref{thm:bipartite}.
\end{proof}

Next, we prove that \DSR\ is PSPACE-complete on planar graphs of maximum degree 6 and bounded bandwidth graphs.
Recall that a graph has \emph{bandwidth} at most $k$ if there exists a numbering $\ell$ of the vertices with distinct integers between 1 and $n$ (where $n$ is the number of vertices of the graph) such that adjacent vertices must have labels at distance at most $k$ (i.e., for every edge $uv \in E$, $|\ell(u) - \ell(v)| \le k$).

\begin{theorem} \label{thm:planar-bandwidth}
\DSR\ is PSPACE-complete on planar graphs of maximum degree 6 and bounded bandwidth graphs.
\end{theorem}

\begin{proof}
First, recall that \VCR\ is PSPACE-complete on planar graphs of maximum degree 3~\cite{Hearn:2005:PSP:1140710.1140715, BONSMA20095215} and on bounded bandwidth graphs~\cite{WROCHNA20181}.
The proof for dominating sets reconfiguration under {\sf TAR} on planar graphs from \cite{Haddadan:2016:CDS:3010167.3010682} works also here since \VCR\ is PSPACE-complete on planar graphs. We use the well-known reduction mentioned in Theorem \ref{thm:bipartite}, which is the following: start with a copy of the original graph $G$ and for each edge $uv$, add a vertex $z_{uv}$ of degree two adjacent to $u$ and $v$. Let $G'$ be the resulting graph, and note that the planarity property of $G'$ is preserved.

Let $G$ be a graph whose bandwidth is bounded by some constant $k$. 
Since a vertex can have at most $k$ neighbors of lower label and $k$ neighbors of higher label, this implies that the maximum degree of $G$ is bounded by $2k$.
We use this observation to prove that the graph $G'$ obtained from the reduction has its bandwidth bounded by $k \cdot (k+1)$. 
We explain how to obtain a labeling $\ell'$ of $G'$ from $\ell$ that satisfies the bandwidth property. The underlying idea is to leave $k$ free values between any two vertices labeled consecutively in the original labeling (i.e., vertices $u$ and $v$ such that $\ell(v) = \ell(u)+1$) in order to label the vertices in $V(G') \setminus V(G)$.

More precisely, for all $i>1$, we relabel the vertex labeled $i$ in $\ell$ with the label $1+(i-1) \cdot (k+1)$ in $\ell'$. Let $u$ and $v$ be two adjacent vertices of $G$ with $\ell(u) < \ell(v)$, then $\ell(v) - \ell(u) \le k$ and thus $\ell'(v)-\ell'(u) \le k\cdot (k+1)$. Moreover, we label the new vertex $z_{uv}$ with label $\ell'(u) + (\ell(v)-\ell(u))$, which lies between $\ell'(u) + 1$ and $\ell'(u) + k$ by the bandwidth hypothesis. We also have $\ell'(v) - \ell'(z_{uv}) < k \cdot (k+1)$, and the bandwidth condition is satisfied.
So the difference between the labels of any two adjacent vertices in $G'$ is at most $k \cdot (k+1)$. 

Observe however that not all vertices have $k$ neighbors of higher label in $G$, and thus the labeling $\ell'$ does not use consecutive values. To fix this, we just relabel the graph 
with values between 1 and $|V(G')|$, maintaining the ordering of $\ell'$. The new labeling $\ell''$ obtained does not increase the distance from $\ell'$, and thus satisfies the bandwidth condition, as required.
\end{proof}

B\"ottcher et al. observed that the \emph{pathwidth} and thus the \emph{treewidth} of a graph are bounded by its bandwidth \cite{BOTTCHER20101217}. Therefore, we immediately get from Theorem \ref{thm:planar-bandwidth} that \DSR\ is PSPACE-complete for bounded pathwidth and bounded treewidth graphs.

\section{Polynomial-time algorithms}\label{sec:poly}
In this section, we focus on graph classes for which \DSR\ can be solved in polynomial time. A natural way to solve this problem is to distinguish a special dominating set (that we call \emph{canonical}) and then show that each dominating set can be reconfigured into this special one \cite{Haddadan:2016:CDS:3010167.3010682}.
The canonical dominating set is not part of the original instance, so it is crucial to be able to compute it in polynomial time if we aim to compute the reconfiguration sequence in polynomial time. However, this is not an issue if we are only interested in the decision problem. We emphasize the fact that this canonical dominating set must be uniquely defined, i.e., the set of vertices that hold a token as well as the number of tokens on each of these vertices must be fixed.

\subsection{Joins and cographs}\label{join-cographs}

In this section, we prove the following theorem, that is of special interest for the case of cographs. Recall that the domination number of a join $G_1+G_2$ is always at most two, since taking a vertex from each operand of the join dominates the whole graph.

\begin{theorem} \label{thm:join} Let $G_1$ and $G_2$ be two graphs, and $D_\source$ and $D_\target$ be two dominating sets of $G_1+G_2$ of the same size. The dominating set $D_\source$ can be reconfigured into $D_\target$ by token sliding if and only if one of the three following conditions holds:
\begin{enumerate}[(i)]
\item $|D_\source| = |D_\target| \ge 3$,
\item the domination number of $G_1$ or of $G_2$ is at most two,
\item both $G_1$ and $G_2$ are connected.
\end{enumerate}
\end{theorem}

\begin{proof}
We first show that if none of these conditions hold, then $(G_1 + G_2, D_\source, D_\target)$ is a {\sf no}-instance. Let $G_1$ and $G_2$ be two graphs with $\gamma(G_1) > 2$ and $\gamma(G_2) > 2$, and assume without loss of generality that $G_1$ is not connected, say with two components $C_1$ and $C_2$. Note that $\gamma(G_1+G_2) = 2$ since neither $G_1$ nor $G_2$ has a universal vertex.

Let $D_\source =\{u,v\}$ and $D_\target = \{w,v\}$ be two minimum dominating sets of $G$ with $u \in C_1$, $w \in C_2$ and $v \in V(G_2)$. We prove that $D_\source$ can not be reconfigured into $D_\target$. Since $G_1$ is not connected, there is no path between $u$ and $w$ in $G[V_1]$. Therefore, the only way to reach $w$ from $u$ is to go through $V(G_2)$. But since $\gamma(G_2) > 2$ no pair of vertices in $G_2$ can dominate $G_2$, and thus no move from $V(G_1)$ to $V(G_2)$ is possible.
\smallskip

We now prove that each of the above conditions is sufficient for the dominating sets to be reconfigured.
\paragraph{Condition (i)}
Suppose $|D_\source| = |D_\target| \ge 3$. Recall that picking a vertex of $G_1$ and one of $G_2$ always forms a dominating set of $G_1 + G_2$. We infer that it is always possible to make one move from $D_\source$ to reach a configuration with tokens in both $G_1$ and $G_2$, then from such position tokens can be slid freely in their part, until reaching $D_\target$ with a last move.

\smallskip

We assume now that $|D_\source| = |D_\target| \le 2$.

\paragraph{Condition (ii)}
For the case when $G_1$ or $G_2$ has domination number at most two, we consider different cases depending on whether a graph has domination number one or not.
\begin{case} \normalfont
If $\gamma(G_1) = 1$ or $\gamma(G_2) = 1$: then $G_1+G_2$ contains a universal vertex. Then, from $D_\source$, one can place a token on this vertex, reconfigure other possible tokens freely, then move that token to reach $D_\target$.
\end{case}

\begin{case} \normalfont
If $\gamma(G_1) = 2$ and $\gamma(G_2) = 2$. Assume without loss of generality that $\gamma(G_1)=2$. Note that in this case, $\gamma(G_1+G_2)=2$, let $ D_\source= \{v_1,v_2\}$. We define an arbitrary canonical dominating set $C$ by taking a vertex (e.g., of smallest index) in each of $G_1$ and $G_2$; we denote these vertices $u_1 \in V(G_1)$ and $u_2 \in V(G_2)$. Recall that each reconfiguration sequence is reversible. Hence, it is sufficient to prove that both $D_\source$ and $D_\target$ can been transformed into $C$. We only show this statement for $D_\source$; the proof for $D_\target$ follows by symmetry.

Suppose first that $v_1$ and $v_2$ belong to the same original graph, say $v_1, v_2 \in V(G_1)$. We show how to reconfigure $D_\source$ into $C$ in at most two steps. First, observe that since $C$ is a dominating set of $G$, $u_1 \in N[\{v_1,v_2\}]$, say $u_1$ belongs to $N[v_1]$. Our first step is to slide the token from $v_2$ to $u_2$, along the corresponding edge of the join. Then, by our observation that $u_1 \in N[v_1]$, we can slide if necessary the token from $v_1$ to $u_1$.

Suppose now that $v_1$ and $v_2$ belong respectively to $V(G_1)$ and $V(G_2)$. Since $\gamma(G_1)=2$, let $\{w_1,w_2\}$ be a dominating set of $G_1$ and thus of $G_1+G_2$ (it can be computed naively in cubic time. It dominates $v_1$ so assume without loss of generality that $v_1w_1$ is an edge. First moving the token from $v_1$ to $w_1$ (if $v_1 \neq w_1$) and then from $v_2$ to $w_2$, at most two steps permit us to reconfigure $D_\source$ into $\{w_1,w_2\}$, which we can then reconfigure into $C$ by the above argument.
\end{case}

\paragraph{Condition (iii)}
Suppose finally that $\gamma(G_1) \ge 3$ and $\gamma(G_2) \ge 3$ but both $G_1$ and $G_2$ are connected. Then $\gamma(G_1+G_2) = 2$ and the minimum dominating sets of $G_1+G_2$ are exactly the sets containing a vertex in $G_1$ and a vertex in $G_2$. Let $D_\source = \{v_1,v_2\}$ and $D_\target = \{w_1,w_2\}$ with $v_1,w_1 \in V(G_1)$ and $v_2,w_2 \in V(G_2)$. Since $G_1$ is connected, there exists a path from $v_1$ to $w_1$ in $G[V(G_1)]$. Moving the token along this path, we always keep a dominating set by the above observation. Doing similarly along a path from $v_2$ to $w_2$, we have a reconfiguration sequence from $D_\source$ to $D_\target$. \qedhere
\end{proof}

We now consider the special case of cographs. Recall that the family of cographs can be defined as the family of graphs with no induced $P_4$, or equivalently by the following recursive definition:

\begin{itemize}
    \item $K_1$ is a cograph;
    \item for $G_1$ and $G_2$ any two cographs, the disjoint union $G_1 \cup G_2$ is a cograph;
    \item for $G_1$ and $G_2$ any two cographs, the join $G_1 + G_2$ is a cograph.
\end{itemize}

Brandelt and Mulder gave in \cite{BANDELT1986182} an alternative characterization of cographs: $G$ is a cograph if and only if $G$ is the disjoint union of distance-hereditary graphs with diameter at most two. Note that computing a minimum dominating set in distance-hereditary graphs is linear-time solvable~\cite{Nicolai01}. Hence, we can compute the domination number of a cograph in linear time as well.

By the previous theorem, we infer that if a cograph is constructed as a join of two cographs, the case is polynomial-time decidable. The case when $G = K_1$, is straightforward. If $G = G_1 \, \cup \, G_2$ is the disjoint union of two cographs, then for two dominating sets $D_\source$ and $D_\target$, deciding whether $D_\source \reconf{{\sf TS}} D_\target$ is equivalent to deciding whether $D_\source\cap V(G_1) \reconf{{\sf TS}} D_\target \cap V(G_1)$ in $G_1$, and $D_\source \cap V(G_2) \reconf{{\sf TS}} D_\target \cap V(G_2)$ in $G_2$, which can be done inductively by induction. As a consequence, we obtain the following:

\begin{theorem}\label{th:cographs}
There is a polynomial-time algorithm deciding \DSR\ in cographs.
\end{theorem}

\subsection{Dually chordal graphs} \label{dually}

\begin{figure}[bt]
    \centering
    \includegraphics[width=0.5\textwidth]{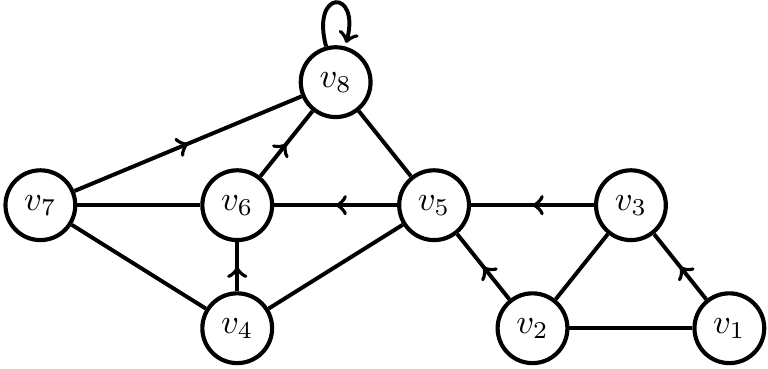}
    \caption{A dually chordal graph.}
    \label{fig:duallyl}
\end{figure}

Let $G=(V,E)$ be a graph with $V = \{v_1, v_2, \ldots, v_n\}$. We denote by $G_i$ the graph $G[\{v_i, v_{i+1}, \ldots, v_n\}]$. A \emph{maximum neighbor} of a vertex $u$ is a vertex $v \in N[u]$ such that we have $N[w] \subseteq N[v]$ for every vertex $w \in N[u]$. In other words, $v$ contains in its closed neighborhood every vertex at distance at most two from $u$. A \emph{maximum neighborhood ordering} (or \mno\ for short) is an ordering of the vertices in such a way that $v_i$ has a maximum neighbor in the graph $G_i$. A graph is \emph{dually chordal} if it has a maximum neighborhood ordering. This ordering can be computed in linear time \cite{BRANDSTADT199843}. Moreover, the \mno\ computed by this algorithm is such that for every vertex $v_i$ (with $i < n$), $v_i's$ maximum neighbor is different from $v_i$ (for connected graphs). An alternative proof of a similar statement for not necessarily connected graphs can be found in \cite{MR3273541}. In the following, we always assume that an \mno\ is associated with a function $mn: V\longrightarrow V$ that associates with each vertex a maximum neighbor.

Note that a dually chordal graph is not necessarily chordal. Figure \ref{fig:duallyl} gives an example of a graph which is dually chordal but not chordal, since it contains an induced cycle on four vertices. The label inside each vertex corresponds to its rank in the ordering, and its maximum neighbor is the endpoint of its single outgoing edge (note that $v_8$'s maximum neighbor is itself). Moreover, observe that any tree $T$ is a dually chordal graph: root the tree in some vertex and orient all edges toward the root; any numbering keeping all $G_i$ connected is an \mno\ where arcs point towards the vertex maximum neighbor.

\paragraph{Link with interval graphs.}
An interval graph is the intersection graph of a family of intervals on the real line. In other words, let \{$I_1, I_2, \ldots, I_n\}$ be a set of intervals. Each interval $I$ can be represented by its extremities $\ell(I), r(I)$ with $\ell(I) \le r(I) \in \mathbb{R}$. We call these values respectively the $\ell$-value and $r$-value (for left and right). The corresponding interval graph $G=(V,E)$ is the following:

\begin{itemize}
    \item $V = \{I_1, I_2, \ldots, I_n\}$;
    \item $I_iI_j \in E \Leftrightarrow I_i \cap I_j \neq \emptyset$, i.e., $\ell(I_j) \le r(I_i)$ and $\ell(I_i) \le r(I_j)$.
\end{itemize}

Let $G=(V,E)$ be an interval graph. For convenience, we denote by $v_i$ the vertex related to the interval $I_i$. We now order the vertices of $G$ with respect to their $r$-values, i.e., $v_i < v_j$ if and only if $r(I_i) < r(I_j)$ (or $r(I_i) = r(I_j)$ and $\ell(I_i) < \ell(I_j)$). Then, we get the following useful property:

\begin{obs} \label{obs:edge}
 Let $v_i$ and $v_j$ be two vertices of $G$ such that $v_i < v_j$. If $v_iv_j \in E$, then for any $v_k$ such that $v_i < v_k < v_j$, we have $v_kv_j \in E$.
\end{obs}

\begin{proof}
Since $v_i < v_k < v_j$, we have $r(I_i) \le r(I_k) \le r(I_j)$. Since $v_iv_j$ is an edge, $\ell(I_j) \le r(I_i)$. Thus, we get that $\ell(I_j) \le r(I_k)$. Adding that $\ell(I_k)\le r(I_k)\le r(I_j)$, the conclusion follows.
\end{proof}

\begin{obs}
Interval graphs are dually chordal graphs.
\end{obs}

\begin{proof}
To see this observation, we prove that the ordering described above is an \mno. For every vertex $v_i$, we show that its neighbor of maximum index $v_j$ in the ordering is a maximum neighbor. Indeed, consider any neighbor $v_k$ of $v_i$ in $G_i$. By definition of $v_j$, we have $v_i < v_k \le v_j$, and $v_k$ is adjacent to $v_j$. Moreover, any other neighbor $v_\ell > v_i$ of $v_k$ either satisfies $v_i < v_\ell < v_j$ or $v_k < v_j < v_\ell$. In both cases, Observation~\ref{obs:edge} concludes the proof.
\end{proof}

An example of the construction used above on an interval graph is given in Figure \ref{fig:interval} where the maximum neighbor of $I_i$ is its only out-neighbor (the endpoint of the only directed edge incident to $I_i$). 

\begin{figure}[bt]
    \centering
    \subfloat[Set of intervals]{
        \centering
        \includegraphics[width=0.7\textwidth]{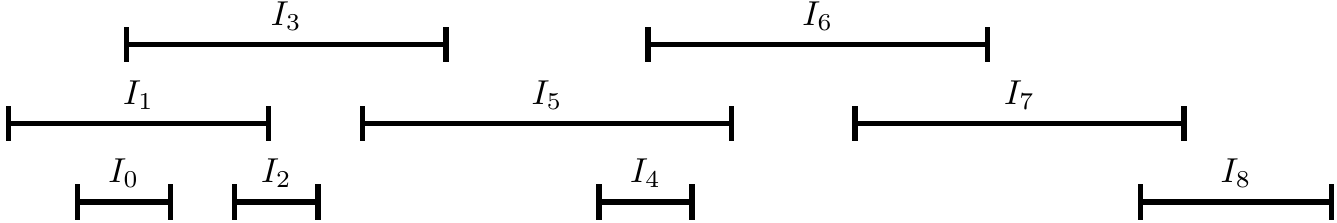}
    }

    \centering
    \subfloat[Corresponding interval graph]{
        \centering
        \includegraphics[width=0.7\textwidth]{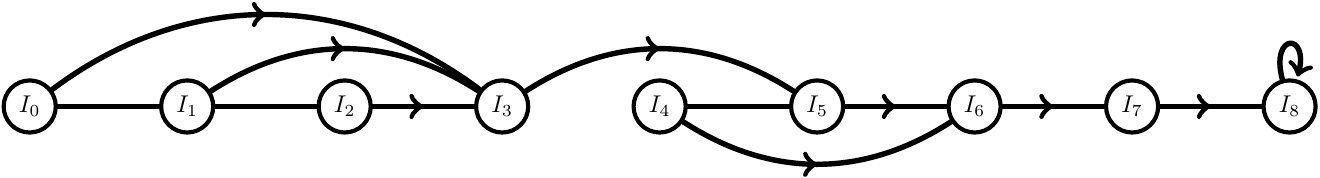}
    }
    \caption{Interval graph and its maximum neighborhood ordering.}
    \label{fig:interval}
\end{figure}

\paragraph{Computing the canonical dominating set.}
Let $G$ be a dually chordal graph, whose vertices are ordered by an \mno. 
Let $C=\{c_1,c_2,\ldots,c_k\}$ be a dominating set of $G$ and $T=\{t_1, t_2,\ldots,t_k\}$ a set of vertices, both sets in increasing order according to the \mno. We say that $C$ is a \emph{triggered dominating set} with triggering vertices $T$ if and only if:
\begin{enumerate}[(i)]
\item $c_i = mn(t_i)$ for all $1\le i \le k$, \label{item1}
\item following the \mno, $t_i$ is the least vertex not in $N[c_1,\ldots,c_{i-1}]$, for all $1\le i \le k$.\label{item2}
\end{enumerate}

It is known that the \textsc{Minimum Dominating Set} problem is linear-time solvable on dually chordal graphs \cite{BRANDSTADT199843}. In our case, we give another algorithm to compute a triggered dominating set, that will serve as a canonical dominating set.

Observe that an \mno\ is associated with exactly one triggered dominating set.
The following algorithm, called \textsc{MDS}, is strongly inspired by the classical algorithm for computing minimum dominating sets in trees~\cite{MITCHELL197976}. It takes as input a dually chordal graph $G=(V,E)$ with an \mno\ and computes a triggered dominating set $C$ of size $\gamma(G)$ and its corresponding set of triggering vertices $T$ in running time $O(|V|+|E|)$.

\begin{algorithm}[ht]
\caption{\label{DS-dually}\textsc{MDS}}
\begin{algorithmic}[1]
\REQUIRE A dually chordal graph $G$ with an \mno.
\ENSURE A minimum triggered dominating set $C$ and its set of triggering vertices $T$.
\STATE Mark all vertices \textsc{Bounded}
\STATE $C \leftarrow \emptyset$
\FORALL {$i$ from 1 to $n$}
\IF {$v_i$ is labeled \textsc{Bounded}}
\STATE label $mn(v_i)$ with \textsc{Required}
\STATE Add $v_i$ to the set of triggering vertices $T$
\FORALL {$u \in N[mn(v_i)]$}
\IF {$u$ is not labeled \textsc{Required}}
\STATE Label $u$ with \textsc{Free}
\ENDIF
\ENDFOR
\ENDIF
\IF {$v_i$ is labeled \textsc{Required}}
\STATE $C \leftarrow C ~ \cup ~ \{v_i\}$
\ENDIF
\ENDFOR
\RETURN $C$ and $T$
\end{algorithmic}
\end{algorithm}

Lemma \ref{lemma:dually-MDS} is devoted to proving the correctness of the algorithm \textsc{MDS}.

\begin{lemma} \label{lemma:dually-MDS}
Given a dually chordal graph $G=(V,E)$, the algorithm \textsc{MDS} computes a triggered dominating set of $G$ of order $\gamma(G)$ in time $O(|V|+|E|)$.
\end{lemma}

\begin{proof}
The fact that $C$ is a triggered dominating set with triggering vertices $T$ is a direct consequence of the construction of the algorithm. Statement \ref{item1} comes from line~5 of the algorithm, while statement~\ref{item2} simply comes from the fact that we deal with the vertices in increasing order in the loop of line~3.
Still we need to prove that this dominating set is of size $\gamma(G)$.

\smallskip

A {\em labeled graph} is a graph whose vertices are labeled \textsc{Free}, \textsc{Required} or \textsc{Bounded}, such that a vertex is labeled \textsc{Free} if and only if it is adjacent to a vertex labeled \textsc{Required} and it is not labeled \textsc{Required}.
Observe that the algorithm \textsc{MDS} maintains all along a labeled graph.
In a labeled graph, we define a \emph{labeled dominating set} a set of vertices containing all the vertices labeled \textsc{Required} and dominating all vertices labeled \textsc{Bounded}. The {\em labeled domination number} is the minimum size of a labeled dominating set. We show that the algorithm \textsc{MDS} keeps the labeled domination number of the graph invariant. At the beginning, when all the vertices are labeled \textsc{Bounded}, the labeled domination number of the graph is exactly its domination number. This will allow us to conclude that the set of vertices marked \textsc{Required} at the end forms a minimum dominating set of $G$, of order $\gamma(G)$.

Let $S$ be a minimum labeled dominating set of a labeled graph $G$, and let $v_i$ be the minimum vertex in the \mno\ that is labeled \textsc{Bounded}. Let $w$ be the maximum neighbor of $v_i$ in $G_i$. If $w\in S$, then $S$ is also a minimum labeled dominating set of the graph $G$ where $w$ is labeled \textsc{Required} and all its neighbors previously labeled \textsc{Bounded} are labeled \textsc{Free}, so the algorithm does not change the labeled domination number of $G$. 
Otherwise, say $v_j$ is the vertex that dominates $v_i$ in $S$. Since $v_i$ is marked \textsc{Bounded}, $v_j$ is not marked \textsc{Required}. If $j \ge i$, then by the maximum neighbor property, $w$ is adjacent to all the neighbors  of $v_j$ that are in $G_i$, so $w$ dominates all neighbors of $v_j$ that are still marked \textsc{Bounded}. Thus we can replace $v_j$ by $w$ in $S$ and keep a minimum labeled dominating set of $G$. This concludes the case $j \ge i$. 
Suppose now $j<i$. Consider $v_k$ the maximum neighbor of $v_j$ in $G_j$. Observe that since no vertex less than $v_i$ is \textsc{Bounded}, by the maximum neighbor definition, $v_k$ dominates any \textsc{Bounded} vertex adjacent to $v_j$. Thus, we can again replace $v_j$ by $v_k$ in $S$ and keep a minimum dominating set. We can iterate until the vertex dominating $v_i$ in $S$ is no less than $v_i$, and then refer to the above argument, used when $j \ge i$. 
This concludes the proof that the algorithm \textsc{MDS} keeps the labeled domination number of the graph invariant, and thus that it produces a minimum dominating set of $G$.

\smallskip

For the time complexity of the algorithm, observe that the algorithm visits every vertex at most once in the main loop, and it visits the neighborhoods of each vertex at most once when it possibly labels it \textsc{Required}. So the complexity is upper bounded by $\sum_{v\in V} O(1 + |N(v)|) = O(|V|+|E|)$.
\end{proof}

\paragraph{The reconfiguration algorithm.} 
We show how to use the canonical triggered dominating set $C$ computed by the \textsc{MDS} algorithm in order to reconfigure two dominating sets of a dually chordal graph. To do that, we provide an algorithm called \textsc{Dually-Chordal-Reconf} that modifies any dominating set $D$ of a dually chordal graph in such a way that $C \subseteq D$. 
The idea of this algorithm is to pick one vertex in $D$ that dominates the triggering vertex $t_i$ (from the output of algorithm~\textsc{MDS}) and to replace it by the corresponding vertex $c_i$ of $C$. Recall that the notation $u\to{{\sf TS}} v$ is used for sliding the token along the edge $uv$.

\begin{algorithm}[bt]
\caption{\label{Dually-Reconf}\textsc{Dually-Chordal-Reconf}}
\begin{algorithmic}[1]
\REQUIRE A dually chordal graph $G=(V,E)$, a minimum dominating set $D$ of $G$
\STATE Compute an \mno\ for $G$
\STATE $(C,T) \leftarrow \textsc{MDS(G)}$.
\FOR {$i$ from 1 to $\gamma(G)$}
\STATE Let $x_i$ be the least vertex of $D \cap N[t_i]$ \label{line-xi}
\IF{$x_ic_i\in E$} 
\STATE $x_i \to{{\sf TS}} c_i$ \label{line-if}
\ELSE
\STATE $y_i \leftarrow mn(x_i)$
\STATE $x_i \to{{\sf TS}} y_i$
\STATE $y_i \to{{\sf TS}} c_i$
\ENDIF
\ENDFOR
\end{algorithmic}
\end{algorithm}

\begin{lemma} \label{lemma:reconf-dually}
Given a dually chordal graph $G=(V,E)$ and a dominating set $D$, \textsc{Dually-Chordal-Reconf} modifies $D$ with respect to the token sliding model in such a way that $C \subseteq D$ in $O(|V|)$ time, where $C$ is the \emph{canonical triggered dominating set} computed by the MDS algorithm.
\end{lemma}

\begin{proof}
Let $T=(t_1, t_2, \ldots, t_{\gamma})$ and $C = \{c_1, c_2, \ldots, c_{\gamma}\}$ be the output of algorithm \textsc{MDS}, with $c_i = mn(t_i)$. 
We denote by $C_i = \{c_1, \ldots, c_i\}$ the set of $i$ first vertices of $C$ according to the \mno. 

In order to prove the correctness of the algorithm, we need to prove that the two following constraints are satisfied:
\begin{enumerate}[(i)]
    \item each move is valid with respect to the token sliding model; \label{propTS-ok}
    \item every intermediate set is a dominating set of $G$. (Note that this ensures the existence of the $x_i$ of line~\ref{line-xi}.) \label{propDomSet-ok}
\end{enumerate}

We prove these two properties by induction on the index $i$ ($0< i \le \gamma$).
For some $i> 0$, assume that the algorithm reconfigured properly $D$ into $D_{i-1}=(D \setminus \{x_1, \ldots, x_{i-1}\}) \cup \{c_1, \ldots, c_{i-1}\}$ .
We explain how to extend this up to rank $i$. By definition, $t_{i}$ is the least vertex which is not dominated by $N[C_{i-1}] = N[\{c_1, \ldots, c_{i-1}\}]$. Let $x_{i}$ be the least vertex dominating $t_{i}$ in $D$. 
Observe that $x_{i} \notin \{c_1, \ldots, c_{i-1}\}$ since $t_{i}$ is the triggering vertex of $c_{i}$. To simplify notation, we denote by $G'$ the subgraph of $G$ induced by vertices larger than $t_i$ in the \mno\ (i.e., the subgraph $G_j$ where $j$ is the index of $t_i$ in the \mno). Note that since $C_{i-1} \subset D_{i-1}$, all vertices in $G\setminus G'$ are dominated. 
 We consider two cases:

\begin{case}
$x_i$ is adjacent to $c_i$.  \normalfont
Observe first that this case occurs whenever $x_i \ge t_i$ in the \mno\ (where $x_i \in N_{G'}[t_i] \subseteq N_{G'}[c_{i}]$). 
In that case, the algorithm executes line~\ref{line-if} and the token sliding constraint~\ref{propTS-ok} is satisfied.
Now, since $c_i = mn(t_i)$ and $x_i$ is adjacent to $t_i$, $N_{G'}[x_i] \subseteq N_{G'}[c_i]$, and constraint~\ref{propDomSet-ok} is also satisfied. The conclusion follows from the fact that all vertices in $G\setminus G'$ are dominated.
\end{case}

\begin{case}
$x_i$ is not adjacent to $c_i$.  \normalfont
This is possibly the case when $x_i <t_i$ in the \mno. The algorithm then first reconfigures $x_i$ into its maximum neighbor $y_i$, which is adjacent to $x_i$ and dominates all neighbors of $x_i$ that might not be dominated yet by $\{c_1, \ldots, c_{i-1}\}$, satisfying constraint~\ref{propDomSet-ok}. 
Moreover, $x_i$ is adjacent to $t_i$ and $x_i < t_i$ in the \mno, so $y_i$, as the maximum neighbor of $x_i$, must be adjacent to  $t_i$ and all its neighbors in $G'$, among which there is $c_i$. So the next move to $c_i$ satisfies the token sliding constraint~\ref{propTS-ok}. 
Also, since $y_i$ is adjacent to $t_i$ and $c_i$ is the maximum neighbor of $t_i$, $N_{G'}(y_i) \subseteq N_{G'}[c_i]$, and constraint~\ref{propDomSet-ok} is also satisfied. This concludes the proof.
\qedhere
\end{case}
\end{proof}

\begin{theorem} \label{thm:dually}
\DSR\ can be solved in quadratic time on dually chordal graphs.
\end{theorem}

\begin{proof}
Let $G=(V,E)$ be a dually chordal graph and $D_\source, D_\target$ be two dominating sets of $G$ of size $k$ (i.e., $(G,D_\source,D_\target)$ is an instance of the \DSR ~problem). Assume that $G$ is connected (otherwise we proceed independently for each connected component, checking first that the number of tokens on each component fit). We explain how to reconfigure $D_\source$ into $D_\target$ in at most quadratic time.

First, we compute in linear time the canonical dominating set $C$ of $G$ with the algorithm \textsc{MDS}. By Lemma \ref{lemma:reconf-dually}, one can transform $D_\source$ and $D_\target$ in such a way that both contain $C$. This can be done in linear time (with respect to the order of $G$) since we move at most $\gamma(G)$ tokens and each move requires at most two steps.  If $k = \gamma(G)$, we are done. Otherwise, choose a vertex $v$ of minimum eccentricity and move all the remaining tokens by a shortest path to $v$. Therefore, the total time complexity is $O(|V|+|E|) + O(|V|) + O((k-\gamma(G)) \cdot \epsilon(v))$, which is at most quadratic (when $k=\Omega(n)$).
\end{proof}

Observe that when $k$ is close to $\gamma$, the algorithm is linear. However, when the number of extra tokens is large (i.e., is linear in $n$), the quadratic overhead may be necessary. Indeed, consider a path on $n$ vertices $P_n$. The minimum eccentricity of $P_n$ is that of the middle vertex $v$ which is $\lfloor n/2 \rfloor$. Therefore, if all the extra tokens are on an extremity of the path, the time needed to move all of them to $v$ is quadratic. Since a path is a dually chordal graph, the conclusion follows.

\section{Open questions}
In all of our polynomial results presented in Section \ref{sec:poly}, computing a minimum dominating set can be done in polynomial or even linear time on the graph classes considered. Therefore, a challenging question is the following: does there exist a graph class for which computing a minimum dominating set is NP-complete but \DSR\ can be solved in polynomial time? As a result of Theorem \ref{thm:dually}, we get that \DSR\ is polynomial-time solvable on interval graphs. Recall that a circle graph is the intersection graph set of chords of a circle. The {\sc Dominating Set Problem} has been shown to be NP-complete on this graph class \cite{KEIL199351}. Hence, we ask the following question:

\begin{question}
What is the complexity of \DSR\ on circle graphs?
\end{question}

If the answer is positive, note that it would generalize our result on cographs. A circular-arc graph is the intersection graph of a set of arcs on the circle. Even if computing a minimum dominating set can be computed in linear time on circular-arc graphs \cite{Chang98efficientalgorithms}, we are interested in the complexity of the reconfiguration version under token sliding. More precisely, we ask for the following:

\begin{question}
Is \DSR\ polynomial-time solvable on circular-arc graphs?
\end{question}

Besides, we found polynomial-time algorithms for cographs and dually chordal graphs but the underlying reconfiguration sequence is most likely not optimal. Indeed, it may be possible that the shortest path in $\mathcal{R}_k(G)$ between the two given solutions does not go through the canonical dominating set. Therefore, can we bound the diameter of the reconfiguration graph? In other words, what is the maximum length of a shortest reconfiguration sequence between any pair of dominating sets? Moreover, what is the complexity of finding the optimal solution, i.e., the shortest reconfiguration sequence between two dominating sets on cographs or dually chordal graphs? Is it polynomial-time solvable, as the reachability variant studied here? Or does it become harder?

\paragraph{Acknowledgements}
We thank the anonymous referees for several remarks which helped improve the presentation of this paper.
This work was supported by ANR project GraphEn (ANR-15-CE40-0009).

\bibliographystyle{plain}
\bibliography{bibliography}

\end{document}